\documentclass[11pt,reqno]{amsart}
\usepackage{amssymb, amscd,enumerate,enumitem,color,
mathtools}

\usepackage[british]{babel}
\usepackage[mathscr]{euscript}
\usepackage{verbatim}

\theoremstyle{plain}
\newtheorem{theo}{Theorem}[section]
\newtheorem{lem}[theo]{Lemma}
\newtheorem{prop}[theo]{Proposition}

\newtheorem{cor}[theo]{Corollary}

\newtheorem{theorem}[theo]{Theorem}

\theoremstyle{definition}

\newtheorem{definition}[theo]{Definition}

\newcommand{\beq}{\begin{equation}}
\newcommand{\eeq}{\end{equation}}

\renewcommand{\d}{\delta}

\newcommand{\g}{\gamma}

\renewcommand{\r}{\rho}
\newcommand{\s}{\sigma}

\newcommand{\D}{\Delta}

\newcommand{\bC}{\mathbb{C}}

\newcommand{\bR}{\mathbb{R}}
\newcommand{\bZ}{\mathbb{Z}}

\newcommand\SL{\mathrm{SL}}
\newcommand\SO{\mathrm{SO}}

\newcommand{\cD}{\mathscr{D}}

\newcommand{\cF}{\mathscr{F}}

\newcommand{\cU}{\mathscr{U}}

\newcommand{\Sim}{\operatorname{Sim}}

\newcommand{\CO}{\operatorname{CO}}

\newcommand{\Gabor}{{\mathcal G}\!\!\operatorname{\it ab}}

\renewcommand{\=}{{:=}}

\newcommand{\Sem}{S^2_-}
\newcommand{\Nem}{S^2_+}
\newcommand{\Np}{{\bf N}}
\newcommand{\Sp}{{\bf S}}

\newcommand{\p}{\partial}

\renewcommand{\square}{\kern1pt\vbox
{\hrule height 0.6pt\hbox{\vrule width 0.6pt\hskip 3pt
\vbox{\vskip 6pt}\hskip 3pt\vrule width 0.6pt}\hrule height0.6pt}\kern1pt}

\DeclareMathOperator\Tr{Tr\;}

\DeclareMathOperator\vol{vol}
\DeclareMathOperator\Id{Id}

\renewcommand\Re{\operatorname{Re}}
\renewcommand\Im{\operatorname{Im}}
\renewcommand\={:=}

\newcommand{\wt}{\widetilde}
\newcommand{\wh}{\widehat}

\newcommand{\bt}{\begin{theo}\ \ }
\newcommand{\et}{\end{theo}}
\newcommand{\bp}{\begin{prop}\ \ }
\newcommand{\ep}{\end{prop}}
\newcommand{\bc}{\begin{cor}\ \ }
\newcommand{\ec}{\end{cor}}
\newcommand{\bl}{\begin{lem}\ \ }
\newcommand{\el}{\end{lem}}
\newcommand{\bd}{\begin{definition}}
\newcommand{\ed}{\end{definition}}

\newcommand{\be}{\begin{equation}}
\newcommand{\ee}{\end{equation}}

\def\<#1,#2>{\langle\,#1,\,#2\,\rangle}

\newcommand{\arr}{\begin{array}{rlll}}
\newcommand{\ea}{\end{array}}
\newcommand{\bea}{\begin{eqnarray}}
\newcommand{\eea}{\end{eqnarray}}
\newcommand{\bean}{\begin{eqnarray*}}
\newcommand{\eean}{\end{eqnarray*}}

\newcommand{\ve}{\varepsilon}

\newcommand{\ster}{\operatorname{st}}

\def\sideremark#1{\ifvmode\leavevmode\fi\vadjust{%            The remark
\vbox to0pt{\hbox to 0pt{\hskip\hsize\hskip1em%               will appear only
\vbox{\hsize3cm\tiny\raggedright\pretolerance10000%          on the side
\noindent #1\hfill}\hss}\vbox to8pt{\vfil}\vss}}}%           in 3cm
\begin{document}

\title[Conformal models for   hypercolumns]{Conformal models for   hypercolumns \\ in  the primary visual cortex V1   
}
%\thanks{D. A.   was partially supported  by  the Grant 
%Basis-Foundation “Leader”  № 22-7-1-34-1
%Grants or other notes
%about the article that should go on the front page should be
%placed here. General acknowledgments should be placed at the end of the article.
%}

%\subtitle{Do you have a subtitle?\\ If so, write it here}

%\titlerunning{Conformal model of hypercolumns}        % if too long for running head

\author[D. V. Alekseevsky and A. Spiro]{Dmitri V. Alekseevsky
 \& Andrea Spiro}

\subjclass[2000]{92B99, 68T45 } %\and MSC code2 \and more}
\keywords{Hypercolumn of the V1 cortex;  
Conformal M\"obius group;   Neurogeometry of the early vision.}

\maketitle
\vskip - 2.5 cm 
\begin{abstract}   We propose  a  differential geometric  model of   hypercolumns  in the primary visual cortex V1 that  combines  features of the   symplectic  model  of  the primary  visual cortex   by   A. Sarti, G. Citti and J. Petitot  and of   the  spherical model  of   hypercolumns by  P. Bressloff and J. Cowan.  The  model  is based on   classical   results   in  Conformal  Geometry. 
%\PACS{87.19.La \and  42.66.Ct. } %\and more}
%
\end{abstract}

%\newpage
%\tableofcontents
\section*{Introduction}
%\addcontentsline{toc}{section}{Introduction}
\label{intro}
%%%%%%%%%%%%%%%%%%%%%%%%%%%%%%%%%%%%%%%%%%%%%%%%%%%%%%%%%%%%%%%%%%%%%%%%%%%%%%%%%%%%%%%%%%%%
 The main purpose of this paper is  to  present     a new differential geometric model for   hypercolumns that   combines  elements of  the  symplectic  model  for  the primary  V1 visual cortex   of   A. Sarti, G. Citti and J. Petitot    and  of the     spherical model  for the  hypercolumns of  P. Bressloff and J. Cowan. We recall that in  Bressloff and Cowan's model  the  simple   cells of a hypercolumn    are   parametrised  by the points of a sphere   $S^2$ with its standard  Riemannian metric $g_o$.  In our model we    consider $S^2$  equipped with the conformal structure  $[g_o]$ determined by  the standard metric and, 
 inspired by   Cartan's approach to conformal geometry,  we   assume that the  simple cells  of a   hypercolumn   are parameterised     by  the points of the    M\"obius  group of conformal transformations of $S^2$. This assumption is similar to  what is considered in     Sarti,  Citti and Petitot's symplectic  model of the V1 cortex,  according to which the simple cells  are parameterised  by  the points of the  conformal  group  of the plane $\bR^2$.  In fact, in the final section of this paper, we   show  that  our model  for hypercolumns  reduces to  (a formal version of)  the  Sarti, Citti and Petitot's model  in sufficiently small neighbourhoods of  pinwheels.
\par
 \smallskip 
  Before going into  the details of our model,  let us first briefly review some   known facts  on the     V1 cortex and its differential geometric  models.
We know that 
 the  firing of a  simple  neuron depends on the value $I(z)$ of   the input energy function $I$  (that is,  the density of the energy   of   the incident light that hits the retina) at   the receptive  field of the neuron,  representable as    a point  $z$ of the retina.  
D. Hubel  and  T. Wiesel    found that  the  firing of a  visual neuron  depend not only on  $I(z)$, but  also  
  on several other data, called {\it internal parameters} (as, for instance,   the {\it orientation} or the  {\it  spatial frequency} of $I$ at a point of a  contour). 
  They also introduced   the important notion of  {\it  hypercolumn}  as {\it a  minimal collection of columns of the V1 cortex  containing    neurons  whose firing depend on any     value of  the internal parameters}.   \par
As it was remarked by  various authors,  
%{\color{blue} as e.g.  W. Hoffman, J. Petitot, Y. Tondut, A. Sarti, G. Citti},   
the existence of internal parameters for the neurons of the V1 cortex  can be formulated  in mathematical language by saying that  the   V1  cortex  may be considered as  a  fiber bundle over the  retina and that the  several  internal parameters   (orientation, spatial  frequency, ocular  dominance,  direction  of  motion, curvature, parameters of  the color  space, etc.),   which  are known to affect the firing of  the visual neurons,  may be considered as   fiber  coordinates for such a bundle
 (see e.g.  \cite{H}; see also \cite{P,P-T,P1,S-C-P}).     N.V. Swindale \cite{S} estimated   the   dimension  of  the fibers of this bundle  (i.e. the number  of  internal parameters) as $6-7$ or $9-10$.\par
  \smallskip
 According  to  Hubel  and Wiesel's results,  the simple   neurons of the  V1  cortex  detect   {\it contours},   that is   the  level  sets  with  large gradient  of the   input   function  $I$ on  retina  $R$ (= the  energy density of the  light  which hits the retina).
 Motivated by the  experimental results  on the structures of the orientation maps,   determined  using Bonh\"offer and Grinvald's revolutionary   techniques  \cite{B-G}, and inspired   by  Hubel and Wiesel's ideas   and the   pioneering     work   by W. C. Hoffman \cite{Hof}, 
 J. Petitot   and  Y. Tondut   and   then  J. Petitot himself in several subsequent papers  and books developed  in great detail    a contact  model for the  V1  cortex \cite{P-T, P}.  In   that  model   the retina $R$  is  mathematically represented as (a region of)  the  Euclidean plane $\bR^2$  and  the   V1 cortex   is  described  as the  total space of the   projectivised tangent  bundle  $\pi: PTR  \to  R$ over  the retina
(\footnote{Note that  Petitot's  model can be easily adjusted into a slightly  more realistic one,  in which the  retina  $R$ is not  represented by a region of the  Euclidean plane $\bR^2$, but  by a region of the  2-sphere $S^2$.}).  This  bundle  admits a natural system of coordinates   $(x,y, \theta)$, in which 
   $(x,y)$ are coordinates for  the points  of the  retina $R = \bR^2$ and   $\theta  \in [-\pi/2$, $\pi/2)$ is the so-called {\it  orientation}, i.e. the  angle   between a line  $\bR v$  of   the  tangent  space $T_{(x,y)} R$    and   the  axis  $0x$. Notice  that   $PTR$  may  be  also interpreted   as the   space of the {\it infinitesimal curves} (or,  more precisely,   the  {\it 1-jets of  the non-parameterised    curves}) {\it of  $R = \bR^2$} and that it  is equipped with   the  canonical contact  1-form $\eta =  dy - \tan \theta\,  dx$. \par
 \smallskip
 We recall that, from a physiological point of view,    the  V1  cortex consists  of  {\it columns},  which  are divided into  {\it regular}  and   {\it singular}.  Each  column  contains  approximately  $80$-$100$ visual neurons, $25\%$ of them  simple ones.  The simple neurons of each column have  almost  the  same  receptive  field  (which we may think  of as a point $z $ of the retina $R$) and  are modelled by   Gabor  filters.  The firing  of a simple neuron, considered as a Gabor filter,    depends 
 not only on its receptive field (RF) $z$, but also on   some 
internal parameters,   first of all  the   {\it orientation} of the contour passing through the RF $z$.  The behaviours of the simple cells of a regular column and  of  the simple cells of a singular column differ by the following aspect.  All Gabor filters of the 
 simple neurons of a {\it regular} column   have the same   orientation $\theta$ (up to an error of  15\%).   This   means that  all  simple neurons of such a regular column  fire  only when  a contour     through  their RF  has    such a   common orientation $\theta$. In contrast with this, a  {\it singular column} (also called  {\it pinwheel}) contains simple cells that are able to  detect contours with any orientation.\par
        \smallskip
 In Petitot's model all columns of the V1 cortex are assumed to be  pinwheels and two pinwheels are   distinguished one from the other  by their   RFs $z = (x,y) \in R \subset \bR^2$ (for a given pinwheel the RFs of its  simple cells are essentially all the same and can be represented by just one point of the retina).  The   simple cells in a fixed  pinwheel  are distinguished   one from the other by   their orientations $\theta \in [-\frac{\pi}{2}, \frac{\pi}{2})$.  In this way the V1 cortex is mathematically represented by the (trivial) bundle $P = R \times S^1 \to R$ with coordinates  $(x, y, \theta)$. 
      \par 
 \smallskip
    Petitot's  contact  model was later  extended   into  the  so-called {\it symplectic model}  in a paper by Sarti, Citti  and Petitot \cite{S-C-P}. In  this second model the simple cells of the V1 cortex are described in terms  of a bundle {\it with  two-dimensional  fibers}. The internal parameters of this model (= the coordinates of the fibers)  are  the orientation  $\theta \in  [-\frac{\pi}{2}, \frac{\pi}{2})$  and a    {\it scaling  factor} $\sigma$ that   describes  the intensity   of    response of a neuron to a  stimulus.
 In mathematical language,  the  simple cells of the V1  cortex  correspond to  the points of the trivial principal bundle  $\pi: P \simeq R \times \CO(2) \to R $ of  the conformal  frames  of  the  retina $R = \bR^2$, with the structure group  $\CO(2) = \bR^+ \times \SO_2 \simeq \bC^*$.      We call this  conformal frame bundle the {\it Sarti-Citti-Petitot (SCP) bundle}.\par
     \smallskip
We think that this very nice model of  Sarti, Citti and Petitot has nevertheless   two weak points.  
 First,   in such a   model,  each  simple cell  is mathematically  represented as a   point of a fiber bundle with two dimensional fiber and each  column is represented  by a  two dimensional fiber, 
 parameterised by   $\theta$ and   $\s$. This means that {\it  the authors consider only 
singular columns}  (for a    regular column,  the  orientation is a fixed number and is not an internal parameter),  {\it while it is  known that in the V1 cortex
most columns are   regular}. Second,   in \cite{S-C-P},  the  authors propose an interpretation of the scaling factor $\s$   in terms of the distance between the    RF  of  a neuron  (which is assumed to be activated through the  so-called {\it maximal selectivity process})   and the regular  boundary of a  retinal figure.  
Such a property has a non-local character and it  therefore  does not fit with  the idea of Hubel and Wiesel that the internal parameters should  correspond to local properties of the input function.
  \par
\smallskip
As we will shortly explain,  the model for  hypercolumns that we propose in this paper yields  a variant of the Sarti,  Citti and Petitot's model for the V1 cortex and   suggests an  alternative physiological interpretations for  the scaling parameter  $\s$   {\it of purely local character} (our model leads naturally to an  identification of $\s$ as the  normalised   spatial frequency) and treat singular and regular columns on the same footing. 
We believe  these two features provide  remedies for  the above   issues.\par
\smallskip
The  starting point for the construction of our model is the observation  that in    both   Petitot's contact model  and  Sarti, Citti and Petitot's   symplectic  model there is the following common assumption: 
Each   simple cell   of the V1 cortex acts  as one of the   Gabor  filters that can be  obtained  from a   {\it mother  Gabor  filter} 
 through the   actions  of the elements of a Lie group $G$ of  transformations (for the contact model,  the group $G$ is the group of isometries of $\bR^2$, for  the symplectic model  $G$ is the group of conformal transformations of $ \bR^2$).  This led  us to   formulate  the following principle, which seem to have been implicitly 
 followed by  Petitot, Sarti and Citti in the formulations  of their models. 
 \begin{itemize}[leftmargin = 10pt]
\item[] {\bf Principle    of Homogeneity}.  {\it A system of visual neurons of the V1 cortex, acting  with  identical  physiological properties, 
can be mathematically  modelled as  a family of  Gabor filters,  that are all obtained  from  a single  {\rm mother  Gabor filter}    by changing the original  {\rm mother profile}  into new profiles via  the  actions  on densities of the  transformations  of  a given  Lie group  (or  pseudogroup)   of  diffeomorphisms. }
\end{itemize}
This principle  guided   us in the  construction of   our  {\it conformal model  for hypercolumns},  which  we shortly present  and which    can be considered as a very  natural modification 
 of  the model 
%
%
%to identify the simple neurons   of a hypercolumn with  the Gabor filters that are obtained from a mother Gabor filter  by  means of the  conformal transformations of  $S^2$.  What we get is a {\it conformal model} for a hypercolumn, which can be  considered as    a modification of the   {\it spherical model}   for hypercolumns
  proposed  by   P.\ Bressloff   and  J.\ Cowan  in  \cite{B-C,B-C1,B-C2}. Before getting into the details of our model, let us  briefly review  Bressloff   and Cowan's model. \par
According to Bressloff and Cowan, each   hypercolumn $H$ of the V1 cortex  is  associated with two pinwheels, say $\Np$ and $\Sp$. The simple neurons of $H$ are determined by two parameters, the orientation $\theta \in [0,2\pi]$ and the (normalised logarithm of the)  spatial frequency $\phi \in [0, \pi]$. Due to this,  $H$ is mathematically  identifiable  with the unit sphere $S^2$, where  $\theta$ and $\phi$ are  the  longitude and the ($\frac{\pi}{2}$-shifted) latitude.  Under this identification, 
 the  pinwheels $\Np$, $\Sp$  correspond to the north and south poles of the sphere, i.e. the two points of the sphere where   the orientation $\theta$ (= longitude)  is not defined  and   $\phi$ takes its maximum or minimum value.
In our conformal model, we  replace the unit sphere $S^2$ by  the  
total space of  the  bundle 
$\pi: P^{(1)} \to S^2$
of the   (second order) conformal frames of $S^2$, obtaining in this way   an extension  of the original Bressloff-Cowan model.  We recall that the bundle $P^{(1)}$  is in turn identifiable with the following homogeneous bundle (see e.g. \cite{Ko}): 
$$\pi: P^{(1)} = \SL_2(\bC) \longrightarrow \SL_2(\bC)/\Sim(\bR^2) = S^2$$ 
 We now remark that 
in a small neighbourhood  of   one of the  two pinwheels, say $\Np$,  our proposed conformal model can  be simplified and transformed into a  second new  model, in which the neurons are parameterised by a smaller number of coordinates. We call it the   {\it reduced model}. According to such a simplified model,  the number of internal parameters is reduced from $4$ to $2$ and the simple neurons of the considered neighbourhood are parameterised by the points of  the total space $\Sim(\bR^2)$ of the SCP bundle $\pi: \Sim(\bR^2)  \to \bR^2$, as it occurs in  the symplectic model of the V1 cortex.  In other words,   {\it the bundle of our conformal model  is  reduced to   a modification of an SCP bundle}. \par
 \smallskip
In the reduced model, the two internal  parameters  for the fibers over  the points   $\Np$, $\Sp$ of the sphere representing the  hypercolumn  
admit an interpretation in terms of   Bressloff and Cowan's  parameters $(\theta, \phi)$. This leads to  an interpretation  of   Sarti, Citti and Petitot's  scaling parameter $\s$  in terms of the  (normalised  logarithm  of the)  spatial frequency $\phi$ considered in Bressloff and Cowan's model, one of the most important internal parameters.  At the same time,  our proposed conformal model offers a mathematical presentation   for the firing of  {\it any}  simple cell of a hypercolumn, regardless whether it belongs to a regular column or to a pinwheel. As we have mentioned above, these are two features that we think  overcome  the above discussed weak points of the Sarti, Citti and Petitot's model.  \par
\smallskip
Summing up,  the problems, which our model addresses, and the solutions, which it offers and from which it  gets its relevance, can be listed as follows: 
\begin{itemize}[leftmargin = 20pt]
\item[(i)] We think that the Sarti-Citti-Petitot modelling space describes the   responses of the V1 cortex as if it were mainly determined  by just the  cells of the  pinwheels. In particular    it does not  clarifies
 the role played by the   regular   columns.  In our conformal model, the  mathematical representation of hypercolumns  can be considered as a combination   of  (a slightly  generalisation of) Bressloff and Cowen's spherical   model    with
  Sarti, Citti and Petitot's model and deals  with  both kinds of  simple  cells (those of the   regular columns and those of the pinwheels)   treating them on the same footing. Since our model \ reduces  to a  modification of   Sarti, Citti and  Petitot's model 
 in  small  neighbourhoods of   pinwheels, we expect  that it might  be  useful to  determine a solid foundation for   a  V1 cortex symplectic model on a large scale. 
 \item[(ii)] The  simplification   of our model  into the reduced  model for the  neighbourhoods of pinwheels makes manifest a possible  origin for  the two internal parameters  of Sarti, Citti and Petitot's model
 and hence  a  physiological interpretation for the  scaling parameter  as an internal parameter, as defined by  Hubel and Wiesel. 
 \item[(iii)] It is  known  that   the   firing   of  the visual neurons in general depend  on many  internal parameters. Recently,  new mathematical  models,      involving  more then  two internal parameters,  have been   introduced  
 (see  for instance  \cite{B-S-C}).  Our  conformal model is one  of such models: Indeed,  our proposed  modelling space is a bundle with  six internal parameters.   It also   very naturally  originates   from  E. Cartan's  approach to  $G$-structures and  to  Conformal Geometry.  We think that these aspects combined with   the  essential role of the conformal group in the remapping problem, as  indicated by Hoffmann \cite{Hof} (see also \cite{A, A-S}), and   its tight relations with   Sarti, Citti and Petitot's  and   Bressloff and Cowan's  models, make our   conformal  model  particularly stimulating. 
 \end{itemize}
\par
    \smallskip
 As a final remark, we would like to point out that,  in   Petitot's contact model,  in Sarti, Citti and Petitot's symplectic model and in our proposed conformal model, the V1 cortex  or a  hypercolumn are  identified with   a principal    bundle   over the retina, whose points can be in turn identifiable   with certain  linear   frames  for the tangent spaces of  a manifold.   More precisely,   in the contact model,  the V1 cortex is identified with  a bundle $P$ with structure group   $G = \SO_2$, which in turn can be  identified   (up to a covering) with  the bundle of the  {\it orthonormal  oriented   frames}   of $R = \bR^2$. In the  symplectic  model,  the  V1 cortex is identified with  a bundle $P$ with structure group   $G = \SO_2 {\cdot} \bR^2$, which can be considered as  the  bundle   of  the  {\it conformal oriented frames} of $R  = \bR^2$.   In our conformal model, a hypercolumn  $H$ is identified with  a principal bundle $P$, which   can be considered as  the bundle of the {\it conformal linear frames} of  the tangent spaces of Bressloff and Cowan's  modelling sphere $S^2$.   We recall that the  principal bundles of linear frames  for the tangent spaces of  manifolds are  the so-called {\it $G$-structures},  and that the  theory of $G$-structures --  introduced and developed  by   S.-S. Chern,  I.\ M.\ Singer,  V.\ Guillemin, S. Sternberg  et al. in the '60 and '70 as an exact  coordinate-free formulation of  classical E. Cartan's method of ``moving frames''   (see e.g.\cite{Ko,St} for excellent introductions) --   is  nowadays  widely used  in Differential Geometry and Physics. The above observations indicates that the  theory of $G$-structures might have  fruitful applications   in neurogeometry as well. Further investigations on this regard would be     very interesting. 
 \par
\medskip

\section{Hoffman's  pioneering model   of  the   V1 cortex} \label{hoffman}
 The  first  attempt  to    develop  a  differential geometric model for  the  primary visual cortex  V1    appeared  in a pioneering  and  very stimulating  paper  by W.  Hoffman  \cite{Hof}. There, among other  important ideas,  for the first time  the crucial  role of the M\"obius conformal group  $\SO_{1,3}$  in descriptions of  the functional structure of the primary visual cortex was pointed out. 
 \par
  Hoffman  proposed a model,  in which       the
 V1  cortex  is  mathematically represented   as  the  $3$-dimensional  total space  $V = \mathcal{C} R$    of a  fiber bundle
  $$  \pi : V = \mathcal{C} R \to  R$$
  over  the  retina  $R$,   whose   fibers  model   the columns  of   the  V1 cortex.  According to this model, the simple neurons of   a column  form an
   {\it orientation response field}  (ORF) and    the collection of all such ORFs 
   determine  lifts of  retinal  contours to  curves in  $V$.  The  tangent  vectors  to  these lifts   define
     a contact structure     on  $V$. 
In Hoffman's words, 
\begin{itemize}[leftmargin = 5pt]
\item[] ``{\it the thing one first thinks of is that the visual contours are integral curves
of the cortical vector field embodied in the ORFs. In other words, the visual
map is a tangent bundle  $  T\mathcal{C}R = V  \to   R$, where    $ R$ denotes the retinal
manifold and V the cortical “manifold of perceptual consciousness.” But the
ORFs unfortunately do line up head to tail as in an Euler line approximation
to an integral curve. Furthermore, the ORFs have an areal character as well
as a line-element one.}''
\end{itemize}
\par
\smallskip
Unfortunately,   in   \cite{H}  diverse mathematical errors and inaccuracies occurred. For instance, it was   stated that   the  V1 cortex  should be identified with    the tangent bundle  $T { R}$ or   cotangent  bundle $T^* { R}$ of the retina,  while on the contrary it is now known that it should be identified with  the projectivisations of these bundles.  It   is also claimed  that the  conformal group $\SO_{1,3}$  acts transitively on the  contact bundle, but this is not true. 
Other   faults of this kind  appeared.  \par
However, despite of these problems, there is no doubt that Hoffmann's   ideas promoted and strongly  influenced  all  subsequent developments  of differential  geometric models for  the visual system.   Following a suggestion of J. Petitot, the area of applied geometry which originated from Hoffman's work  is nowadays  called {\it neurogeometry}.\par
\medskip
\section{Petitot's  contact model of the V1 cortex}\label{petitot}
  A   fully correct mathematical    formulation of Hoffman's  ideas   was  given  by    J. Petitot  and Y. Tondut in \cite{P-T} and by J. Petitot  in \cite{P}. 
 The  starting point  was to consider  the V1  cortex  as  a  bundle over  the  retina  with  one dimensional fiber, parameterised by   the  most important internal parameter, the  {\it  orientation}.
  More precisely,    Petitot  and Tondut identify  the  retina  $R$   with the Euclidean plane $R = \bR^2$  with Euclidean coordinates $(x,y)$
 and propose  to represent  the V1  cortex   by  the projectivised cotangent   bundle
 $$\pi:   PT^*(\bR^2)= S^1 \times \bR^2  \longrightarrow  R=\bR^2\ .$$
  In this bundle   the  fiber $ S_z^1 =\pi^{-1}(z)$ at a point   $z=(x,y)$  is  the projective line  (that is, a circle) $P T_z(\bR^2)\simeq  \bR P^1  = S^1$, endowed  with the  natural coordinate   $ \theta  \in [0, \pi) $  (= the orientation) given by  the    angle between each line  $[v] \in PT_z(\bR^2)$   and the  $x$-axis.
  In   the  model, each   fiber  $\pi^{-1}(z)$ of $P T(\bR^2)$   corresponds to  a   pinwheel with RF $z$ and
   the simple neurons  of such pinwheel  are parametrised by  their  orientations $\theta$.
  Since the  standard   Riemannian metric  of the  sphere   is conformally flat  and  all   main constructions  are  conformally invariant,   Petitot 's model can be easily generalised to a setting   in which  the  retina $R$ is identified  with  (an open domain of) a  sphere.
  \par
  \smallskip
  Let  us denote by   $I(z) =  I(x,y)$ the  {\it input  function of the  retina},  that is,  roughly speaking,   the density function of  the  energy of   the light, which is incident to the  retina.  The {\it  retinal contours}   are  the  level sets of the form $C = \{I(x,y)= \text{const} \}$ on which   the gradient of  $I$  is large. According to Hubel and Wiesel,  they are  the  main objects   of perception in  the early vision.\par
  A retinal   contour $C$ is a   non-parameterised  curve, but   one can  always  locally     select   a   parametrisation  $z(t) = (x(t), y(t))$, $t \in (a,b)$  with  $\dot{x}(t) \neq 0$   and then   consider  $x$ as a new parameter so that     the  contour $C$  can be   represented as  a curve  $z(x)=(x, y(x))$.
  The orientation  $\theta(x) \in [- \pi/2, \pi/2) $ of the curve $ z(x)$  at  a point $o=(x_o,y_o))$   is  the  angle  between the  coordinate line  $x$ and the  velocity  $d z(x)/dx$ at that point.      This implies that  the  curve  $ z(x)=  (x,y(x))) $  is a solution to   the  ODE
     \begin{equation}\label{ODE}
     \frac{dy(x)}{dx} = p(x):=\tan \theta(x).
     \end{equation}
       Given a  pinwheel  $S^1_{z_o} = \pi^{-1}(z_o) = P T(\bR^2)|_{z_o} $ with RF in $z_o$,  the simple neuron  in  $S^1_{z_o}$  with  orientation $\theta_o$   fires   when  there is a   contour $C= z(x)$ passing  through $z_o = (x_o, y_o)$ and  with  orientation  $ \theta(x_o) = \theta_o$  at such a   point.  It follows that  the  curve $c(x) =  (x,y(x), \theta(x))$ in $P T^* \bR^2$, made  of    the fired   neurons,  is a lift   to  the V1 cortex $PT^* \bR^2$ of the  contour $z(x) = (x, y(x))$. \par
       \smallskip
        The   lifted  curve $ c(x) = (x,y(x), \theta(x))$ is  {\it horizontal} in the sense that, at each of its points,   it is tangent to the canonical  contact distribution  $\cD_H \subset T (P T^* \bR^2)$ of $P T^* \bR^2$. Indeed, if we denote by $  \lambda_H  = d \theta - \tan \theta dx$  the canonical contact form of $P T\bR^2$ and by $\cD_H = \ker \lambda_A$ the corresponding contact distribution, at any point of  the lifted curve $c(x)$  the   velocity is
       \begin{equation} \label{111}  \dot c(x) =\p_x + \frac{dy}{dx}\p_y + \dot \theta(x)\p_{\theta} \in T (P T\bR^2) \simeq  T H\ ,\end{equation}
           which implies that
       $$ \lambda_H (\dot c(x)) = (dy -  \tan \theta dx)(\dot c(x))= \frac{dy}{dx}-p(x) =0\ \qquad \Longrightarrow\qquad \dot c \in \cD_H|_{c(x)}\ .$$
 Such  horizontal lifts to   $P T^* \bR^2$ of  curves  on $\bR^2$ are the so-called  
 {\it Legendrian lifts}. Hence, according to Petitot and Tondut,  the  V1  cortex  realises the  Legendrian  lifts of the contours  in  the retina  and the  aim of the visual system  is   to  integrate the   infinitesimal information encoded in  the  firing of the  simple neurons 
 (i.e. to integrate the   ODE \eqref{111})   and globally  recover  the   contours.\par
\smallskip
We stress the fact that the  most important retinal contours  are  the  {\it closed}  ones,    corresponding  to  the boundaries of the  retinal images  of   three dimensional  objects. A closed contour   divides the  retinal plane into two parts, one of which is the retinal  image of  an external object.  Deciding which of the two parts of the plane is the ``image''  is equivalent  to fix an orientation of the  closed contour $C$ (it is the orientation, according to which the image is  to the left of $C$ with respect to direction of such orientation).  A very important  problem  is to   understand  under which mechanism  the   visual system  chooses the  orientations of the closed contours, i.e.    which   of  the two  regions  is the image of  a three dimensional object.\par
 \medskip
\section{Sarti, Citti and Petitot's symplectic  model}\label{sarticittipetitot}
In  \cite{S-C-P},  Sarti, Citti and Petitot proposed   a symplectisation of  Petitot's  contact model.
They  assumed  that   each    simple  cell  of the V1 cortex  is  characterised  not only by an orientation  $\theta$ and
   by  the  point  $z \in R$ of the retina, corresponding to  the {\it receptive field}  (RF) of the simple cell, but also by a new  parameter $\s$,  called
{\it scaling}. It represents  the intensity of    the  reply  to a   stimulus.
    This assumption   leads to  the extension of contact  model  into   the so-called  {\it symplectic  model}, described as follows.
According to this  model,
 the  V1 cortex is  represented by  the principle $\bC^*$-bundle
 over   the  retina $R = \bR^2$
$$\pi : P = \bC^* \times R  \longrightarrow R\ ,$$
 where $\bC^* = S^1 \times \bR^+ = \{\s e^{i \theta}  \}$ is
the   group of non zero  complex numbers.
The  $\bC^*$-bundle $P$ is  also identified with the  group $P = G$ of similarities of $\bR^2$
 $$  G:= \Sim(\bR^2) =\left(\bR^+ {\cdot} \SO_2\right) \ltimes T_{\bR^2}  = \bC^* \ltimes \bR^2\ ,$$
  where  $T_{\bR^2} = T_{\bC}$ denotes the group of  the parallel translations of the plane. The same bundle   can be  also  identified    with the  cotangent  bundle  of $R$  with the   zero section removed  $T^*_{\sharp} R$.
 The  manifold    $P  =  T^*_{\sharp}R $   has a natural  symplectic structure, given  by the non-degenerate closed $2$-form
  $\omega = d \lambda$ determined by   the Liouville form $\lambda$.
    We  call    the symplectic bundle $\pi: P  = G =  T_{\sharp}^* R \to R$   the {\it Sarti-Citti-Petitot (SCP) bundle}.\par
  \smallskip
The  correspondence  between    visual neurons   and  points of  the SCP bundle, on which the symplectic model is built,  consists 
of the identification of  each simple cell  of  the V1 cortex   with the linear filters for  the input function  of the retina defined as follows.  We recall that   a linear filter   on the input energy function $I(z)$
 is a map 
 $$I  \overset{T_W}\longmapsto \int_D W(z)I(z) \vol\ ,\qquad   z=(x,y)\ ,$$
 where we denote  by  $\vol :=  dx\,dy$ and  by $W(z)$   a  {\it  density}   with support $D$  ($W$   is called  {\it receptive
 profile (RP)} or {\it weight} and $D$ {\it receptive field (RF)}). Roughly speaking the linear functional     computes a  sort  of  ``mean value'' of the restriction
$I\vert_D$ of the input function to the receptive field $D \subset R$,  where  each   point   $z$  is counted
with  the  weight    $W(z)$.  The   {\it ``mother'' Gabor    filters} are the two  linear functionals of the above kind in which the  RP $W(z)$  is of either one of the following two forms 
       $$  \g^+_0(z):= \g_0(z)\cos y = e^{-\frac{1}{2}\vert z\vert ^2} \cos y\ ,\qquad  \g^-_0(z):= \g_0(z)\sin y    = e^{-\frac{1}{2}\vert z\vert ^2} \sin y\ ,   $$
These two RPs  conveniently combine into  the   a single  {\it complex} RP
        \beq \label{mothergaborcomplex}  \g_0^{\bC}(z) := \g_0(z) e^{iy} = e^{-\frac{1}{2} \vert z\vert ^2 +iy}  =   \g^+_0(z) + i  \g^-_0(z) \ ,  \eeq
  which is called {\it RP of  the complex mother  Gabor  filter}  $\Gabor_{\gamma_0^\bC} $.  An  {\it even} (resp. {\it odd}) {\it Gabor filter} is   a linear filter $T_W$
  with RP profile  $W(z)$ given by  the first (resp. the second)  of the  following densities
  $$\g^+_g(z) = \Re(\gamma^{\bC}_0(g(z)))\ ,\qquad \g^-_g(z) = \Im(\gamma^{\bC}_0(g(z)))$$
  where  we denote by   $g$  a transformation in  $G =\Sim(\bR^2)$.   The main assumption of the SCP model is the following: 
  {\it  each simple cell of the V1 cortex  works as  one of the above defined Gabor filters, which are in turn in  bijection with the points of  $G =\Sim(\bR^2)$, the latter   
  naturally identified with the total space of the above defined  SCP bundle}.
  \par
\medskip
   \section{Bressloff and Cowan's  spherical model\\ for  the hypercolumns  of    the V1  cortex }\label{sect8}
      Bressloff   and   Cowan \cite{B-C,B-C1,B-C2} proposed   a Riemannian  spherical  model for the hypercolumns, which is  based on two   parameters: the orientation $\theta$  and  the   spatial  frequency $p$. They   assumed  that a  hypercolumn $H$ is a domain in the V1 cortex,     associated  with two  pinwheels $\Sp,\Np$,       corresponding   to the minimal  and   maximal values  $p_-$,  $p_+$ of the spatial frequency.  According to such a  model, the simple  neurons  of the hypercolumn $H$ are  parametrised  by  their orientation $\theta$  and  their {\it normalised  spatial frequency } $\phi$,  given by
       $$\phi =\pi \frac{\log(p/p_-)}{\log(p_+/p_-)}  \ .$$
   The definition of $\phi$  is done in such a way  that   it  runs exactly between     $0 $ and $ \pi$.
For any choice of  the parameters  $\theta$ and $\phi$, the corresponding  simple  neuron  $ n = n(\theta, \phi)$  fires  only  if  a  stimulus occurs  in its RF   with      orientation  $\theta $  and with    normalised  spatial  frequency   $\phi$. \par
     \smallskip
   Bressloff and Cowan proposed to  consider  the parameters $\theta$ and $\phi$   as  spherical coordinates,  where      $\theta \in [0, 2\pi)$ corresponds to  the {\it longitude} and $\phi = \phi' + \pi/2 \in [0, \pi]$  corresponds to  the  {\it polar  angle} or  {\it shifted latitude}  (with  $\phi' = \phi- \pi/2$    {\it latitude} in the usual sense).      The shifted  latitudes of the pinwheels  $\Sp,\Np$  are $0, \pi$,  but the   longitude (=orientation) is not defined for them -- in fact, the pinwheels  are able to detect contours of {\it any}  orientation.  This led Bressloff and Cowan to  identify the hypercolumn $H$  with  the sphere $H_{BC} =S^2$,  equipped  with the  spherical coordinates $\theta$,  $\phi$. The  two pinwheels $\Sp, \Np$ of $H$ are identified with   the south  and the north pole  of the sphere $H_{BC} = S^2$.\par
\par
\smallskip
 The   union  of the RF's  of  the neurons of a  column (resp.  hypercolumn)  is called  {\it the  RF of the column} (resp. {\it of the  hypercolumn)}.  Assume that  the RF of the simple neurons are  small, so that they can be considered just as   points.  We therefore get that  there is   a map  from the sphere  $H_{BC} = S^2$   into  a corresponding region  
 $R_H \subset R$ of the retina (that represents the RF of the hypercolumn  $H$): 
 $$\pi: H_{BC} = S^2 \longrightarrow R_H\ ,\qquad n \longmapsto z(n) = \text{receptive fields of}\  n\ . $$ 
 According to  these ideas , the projection
 $$\pi: H_{BC} =S^2 \setminus \{\Sp, \Np\} \longrightarrow  R_H \setminus \{\text{receptive fields of $\Sp$ and $\Np$}\}$$
 is a (local) diffeomorphism between  the following two objects: (1) the hypercolumn  with the pinwheels $\Sp, \Np$ removed  and (2) the   region $R_H \subset R$ consisting  of the corresponding receptive fields.  This is tantamount
 to assume that:
 \begin{itemize}[leftmargin= 20pt]
 \item[(a)] the simple neurons of a  regular column in  $H_{BC} = S^2  \setminus \{\Sp, \Np\}$    with   equal  orientation $\theta$  (i.e. the collection of the simple neurons of $H$  that  fire for the same    orientation $\theta$) are distinguished one from the other by  their   spatial frequencies $\phi$;
 \item[(b)] the   simple neurons  in (a)  have  different RFs   if they have different   spatial frequencies, and  they are no longer assumed to be functionally equal.
 \end{itemize}
  \par
\medskip
\section{A short  introduction to the  conformal geometry of the sphere}\label{sect8*}
In this subsection, we   briefly recall  a few facts  of  the theory of  conformal transformations of the $2$-sphere  (for more detailed  introductions and   references, see e.g. \cite{K,A-G,P-S,Su}).\par
 \subsection{The Riemannian   spinor  model of  the conformal  sphere} \label{sect51}
Let us  identify  the   sphere $S^2$  with the {\it Riemann  sphere}, i.e. with  the one-point compactification of the complex plane $\bC = \bR^2$
$$S^2 = \wh{\bC} = \bC \cup \{\infty \}\ ,$$
with   the south and north  poles $\Sp$, $\Np$ of $S^2$ identified with   $\Sp=0$ and  $\Np = \infty $, respectively,  and equipped  with the  standard  complex  coordinates
$z$  for    $S^2\setminus \{\Np\}  = \bC$  and  $w \= \frac{1}{z}$ for   $S^2\setminus\{\Sp\} = (\bC \setminus \{0\}) \cup \{\infty\}$. 
  \par
       \smallskip
       The  group  $ \SL_2(\bC)$  acts    on
       $S^2 = \wh \bC$ as   the conformal  group  of  the   linear  fractional    transformations
    $$ z \longmapsto Az = \frac{az +b}{cz +d}, \qquad   a,b,c,d \in \bC,\ \qquad   \det A =  \begin{pmatrix}
a & b \\
c  & d
\end{pmatrix}  =1\ ,$$
      with    kernel $\bZ_2 = \{\pm \text{Id}_{\wh \bC}\}$. 
      Consider  the subgroups of $\SL_2(\bC)$  defined by    
\beq\label{gd}  \begin{split} & \hskip4  cm  \bC^* = \left\{ \begin{pmatrix}
a&  0\\
0& a^{-1}  \end{pmatrix}, \, a\in \bC\setminus\{0\} \right\} \\
& N^+  =  \left\{  A   =
\begin{pmatrix}
1&  b\\
0& 1 \end{pmatrix}\ ,\ b \in \bC\ \right\}\qquad \text{and}\qquad N^-  =  \left\{  A   =
\begin{pmatrix}
1&  0\\
c& 1 \end{pmatrix}\ ,\ c \in \bC\ \right\}\ , 
\end{split}
\eeq
i.e.  the  diagonal subgroup and the   unipotent  subgroups of  upper  and  low  triangular  matrices of $ \SL_2(\bC)$, respectively.
Let us also denote by  
$$B^- \=  N^- {\cdot} \bC^*  \simeq  \CO(2) {\cdot} \bR^2 \ ,\qquad B^+  \= \bC^*  N^+  \simeq\CO(2) {\cdot} \bR^2$$
 the   Borel     subgroups of   $\SL_2(\bC)$ given by 
    the  lower  triangular and the upper  triangular    matrices, respectively.
Notice that for  the  open  and dense   subset   $\cU = \big\{ A =\left( \smallmatrix a & b\\ c & d \endsmallmatrix\right) \ :\ a \neq 0\big\}  \subset \SL_2(\bC)$,  each $A \in \cU$   decomposes into  the product 
\begin{multline*} A = A^-  A^0 A^+\ ,\\
 \text{with}\   A^- = \left(\begin{array}{cc} 1 & 0\\ \frac{c}{a} & 1 \end{array} \right)\in N^{-} \ ,\     A^0  = \left(\begin{array}{cc} a &0 \\ 0 & \frac{1}{a} \end{array} \right)\in \bC^* \ ,  \ A^+ = \left(\begin{array}{cc} 1 & \frac{b}{a}\\ 0 & 1 \end{array} \right)\in N^{+}\ ,\
\end{multline*}
 so that    the open dense subset  $\cU \subset \SL_2(\bC)$  admits the  {\it Gauss decomposition}  
$ \cU =  N^- {\cdot} \bC^* {\cdot} N^+$.\par
  \smallskip
  The   stability  subgroups $(\SL_2(\bC))_{\Sp}, (\SL_2(\bC))_{\Np} \subset \SL_2(\bC)$ of  the   points $\Sp= 0$  and  $\Np= \infty$
are    
$$(\SL_2(\bC))_{\Sp} = B_- \ , \qquad \  (\SL_2(\bC))_{\Np}  = B_+\ \ ,$$
and are  both isomorphic to $\Sim(\bR^2) = \CO(2) {\cdot} \bR^2$.
Thus,  as  a  homogeneous  manifold,   $S^2 = \wh{\bC}$ is   identified  with the   coset  space
$$S^2 =  \SL_2(\bC) /B_{\mp} =   \SL_2(\bC)/ \Sim(\bR^2)\ .$$
\par
\smallskip
\subsection{The  stereographic  projections  of  the   sphere}
 Let us now identify  $S^2 $ with  the standard unit sphere  $S^2 = \{ x^2 + y^2 + z^2 = 1\}  $ in $ \bR^3$.
       The  {\it stereographic  projection $\ster_{\Np}$ of  $S^2$    from  the north pole $\Np$   onto  the  tangent  plane  $\Pi$  at the  south  pole $\Sp$} is the $ (\SL_2(\bC))_{\Np}$-equivariant  conformal mapping
     $$\ster_{\Np}  : S^2 \longrightarrow \Pi := T_{\Sp} S^2$$
 that   sends   each   point $A$ of the punctured sphere  $S^2 \setminus \{\Np\}$  to the point of the  intersection $\overline A$ between the tangent plane $\Pi$ and  the  ray   $\ell_{\Np A} \subset \bR^3$  with origin $\Np$ and  passing through $A$.
 The map $\ster_{\Np}$ naturally defines complex coordinates on the sphere $S^2$  as follows. Consider coordinates $(X, Y, Z)$ so that  that   $S^2 \subset \bR^3$ is defined  by the equation
$$X^2 + Y^2 + (Z-1/2)^2 =1/4$$
so  that $\Sp = (0,0,0)$ in  the coordinates $(X, Y, Z)$. Then  the  tangent   plane
   $ \Pi = T_{\Sp} S^2$  is the     plane $Z=0$ with  induced  coordinates $(X,Y)$.  If we identify   $\Pi $  with    the complex line $\Pi = \bC$ with complex  coordinate  $z = x +i y$, the stereographic  projection  takes the form
  $$ \ster_{\Np} :  S^2 \ni A = (X,Y,Z) \longmapsto z= x+iy :=  \frac{X+iY}{1-Z} $$
  and its inverse map is
   $$ \ster^{-1}_{\Np} : \Pi \ni z= x+iy \longmapsto(X,Y,Z)=\frac{1}{1 + |z|^2} \big(x, y , |z|^2\big). $$
   This expression allows to consider $z = x + i y$ as a holomorphic  coordinate  on   $S^2\setminus \{\Np\} \subset \bR^3$. \par
    A similar definition gives 
    the  {\it stereographic  projection} $\ster_{\Sp}: S^2 \setminus \{\Sp\} \longrightarrow  \Pi' \= T_{\Np} S^2$  from the south pole. It determines a holomorphic coordinate $w$ on $S^2 \setminus \{\Sp\}$, which  satisfies  $w = \frac{1}{z}$ on $S^2 \setminus\{\Np, \Sp\}$. This 
     motivates  the  previous  identification  of $S^2$ with  the Riemann sphere
    $S^2 = \wh {\bC} = \bC \cup \{\infty \}$.\par
   \smallskip
   In  the   holomorphic coordinates $z$ and $w = \frac{1}{z}$,  the  conformal action   of  the  stability   subgroup  $ (\SL_2(\bC))_{\Np}= B_+ \simeq \Sim(\bR^2)$ on $S^2$  becomes   affine:
    $$(\SL_2(\bC))_{\Np} \ni A =
    \begin{pmatrix}
    a&b\\
    0&a^{-1}
    \end{pmatrix}
    : z \longmapsto A z =  a'z +b'\qquad \text{where}\ \ \   a' :=a^2,\ \  b' := ab.$$
    and  the   lower  triangular nilpotent  group  $N^- $   acts   as
    $$
    N^-  \ni C =
    \begin{pmatrix}
    1&0\\
    c& 1
    \end{pmatrix}
    : z \longmapsto C(z) = \frac{z}{1+cz} = z(1 - cz + (cz)^2 - \ldots  ) $$
We remark that  $N^-$ acts as   the group of  linear fractional transformations  which is generated  by the  holomorphic  vector  field  $ z^2 \partial_z$ and
 acts  trivially   on the tangent  plane $\Pi = T_{\Sp} S^2$.
\par
\smallskip
  \subsection{Basics of   conformal geometry of the sphere}\label{fund}
The {\it conformal  structure}  of a Riemannian  manifold $(M, g)$ is    the   class   $[g] = \{ \lambda g\  : \   0< \lambda  \in C^{\infty}(M)\}$ of  all metrics that are  conformally  equivalent to $g$. The  {\it conformal  group   of $(M, [g])$} is  the   group     $ \operatorname{Conf}(M)$  of the   transformations,   which   preserve  $[g]$. The {\it conformal sphere} is the unit sphere $S^2 \subset \bR^3$  equipped with the conformal structure $[g_o]$ of  the standard round metric  $g_o$.  It is known that  $\operatorname{Conf}(S^2)$ has two connected components.  The connected component of the identity  is   
$$\operatorname{Conf}^o(S^2) \simeq \SO^o_{1,3} \simeq \SL_2(\bC)/\bZ_2\ .$$  All  stability subgroups at points $p \in S^2$ are  isomorphic each other and each of them is  isomorphic to $ (\SL_2(\bC))_{\Sp} =   = B_- \simeq \bR^+ {\cdot} (\SO_2 \ltimes \bR^2)   = \Sim(\bR^2)$,   the group of similarities of $\bR^2$. 
\par
The  above defined stereographic  projection
$     \ster_{\Np} :  S^2\setminus  \{\Np\} \longrightarrow  T_{\Sp} S^2 = \bR^2 = \bC$
naturally extends  to a  conformal diffeomorphism between the sphere $S^2$ and the {\it Riemann sphere} 
$$\varphi :  S^2  \longrightarrow   \wh{\bC} = \bC \cup \{\infty\} $$
 mapping $\Np$ to $\infty$.
As we mentioned in \S \ref{sect51},  the group $ \SL_2(\bC)$ acts  transitively on   $\wh{\bC}$ as  the   group of  linear fractional transformations with  kernel $\bZ_2 = \{\pm \Id\}$ and  its  stability  subgroup $(\SL_2(\bC))_{p = \infty}$  at 
  $\infty $      is the   group   $\operatorname{Sim}(\bR^2) =  \{z \mapsto az+b\} $ of  the  similarity transformations of the plane (which, by
Liouville Theorem,    is  also   the  connected group of  conformal transformations of  $\bR^2$).  The diffeomorphism $\varphi$ is $G/\bZ_2$-equivariant and determines an identification between 
 the Riemann  sphere  $\wh{\bC}$  and  the conformal  sphere   $S^2  = G/ \operatorname{Sim}(\bR^2)$.
\par
 \smallskip
 Denote  by  $\cF^o =  (0,  (f_1^o, f_2^o))$   the  standard  frame   of $\bR^2 = \bC$  at the origin.
 An {\it oriented  conformal    frame   at a   point  $z \in \bC =  \bR^2$} is an  orthogonal  frame  $\cF =(z, (f_1,f_2))$  with   the  same   orientation  as    $\cF^o$   and with  the vectors $f_1,   f_2$   of  the  same length.     Any   conformal  frame  $\cF$ at  a   point  $z \in \bR^2=\bC  $   can  be identified with  the  1-jet  $j^1_0(h)$  of the unique  conformal transformations $h  \in  \operatorname{Conf}(\bR^2) = \operatorname{Sim}(\bR^2)$  that  maps the origin   $0$  into  $z = h(0)$ and the frame  $\cF^o$ into $\cF$.
Like in the  cases of the isometry  groups and  bundles of orthonormal frames, the   similarity  group $\operatorname{Sim}(\bR^2)$ acts   simply transitively on the total space of the   bundle    $\mathcal{CF}r(\bR^2) \to \bR^2 $ of all   oriented    conformal  frames  of $\bR^2$ and   there exists a 
$(\operatorname{Sim}(\bR^2))_0$-equivariant  diffeomorphism
$$     \operatorname{Sim} (\bR^2) \ni  h \longmapsto      \cF_h  \=   \bigg(h(0),  (h_* f_1^o, h_*f_2^o)\bigg) \in  \mathcal{CF}r(\bR^2) \ . $$
 This construction   can be   generalised  to the     case of  the conformal sphere as follows.       
We   define   as  {\it second order   conformal frame   at  a point  $p$}   the   $2$-jet  of the form  $j^2_0 (\varphi^{-1}\circ \wh h)$,   where $\varphi^{-1}$ is the inverse stereographic projection $\varphi^{-1} : \bR^2  \to S^2 \setminus\{\Np\}$  (which is also equal to  the inverse of a local conformal coordinate system around  the south pole $\Sp = (0, 0,-1)  \in S^2$)
and  $\wh h$ is  a transformation in $ \operatorname{Conf}^o(\wh \bC)$  such that  $ h(\Sp)   =p$ (so that   $\varphi^{-1}\circ \wh h$ is  the inverse of a local conformal coordinate system  around $p$).   \par
The manifold  $\mathcal{CF}r^{(2)}(S^2)$ is  the   total space of the principal  Cartan bundle of second order conformal frames
 $$   \pi:    \mathcal{CF}r^{(2)}(S^2)   \longrightarrow S^2   $$
 with   the     group $ \operatorname{Sim}(\bR^2)=( \operatorname{Conf}^o (S^2 ))_{\Sp}$,   as  structure group. 
Since the conformal group $ \operatorname{Conf}^o (S^2)$  acts  simply transitively   on     $\mathcal{CF}r^{(2)}(S^2)$ as  a  group  of  automorphisms,   the Cartan   bundle  $\pi:  \mathcal{CF}r^{(2)}(S^2)   \longrightarrow S^2$  of second order   frames  is  identified with the bundle 
$$   \pi : \SL_2(\bC) \longrightarrow     \SL_2(\bC)/\operatorname{Sim}(\bR^2)=      \SL_2(\bC)/B_-  =S^2 = \wh{\bC}  $$
  associated with the  homogeneous space   $S^2 =    \SL_2(\bC)//B_-$ (see \S \ref{sect51} for notation).   
Since the left  translation   defines an  absolute parallelism on the Lie group $G =   \SL_2(\bC)$  (i.e.  a canonical  identification between  all tangents spaces $T_gG$, $g \in G$), the identification  $\mathcal{CF}r^{(2)}(S^2) = \SL_2(\bC)$  gives  a canonical {\it Cartan connection} on the principal bundle $\pi$,  playing a  crucial  role   in   conformal geometry.   In particular,  it  allows to generalise the classical Frenet theory of the curves in $\bR^2$ to the conformal case (see e.g. \cite{Su}) and    determines  the conformal invariants of $S^2$.
 \par
 \medskip
  \section{The conformal  model   and the reduced   model  of a   hypercolumn} \label{THEMODEL}
\subsection{The conformal  model as an extension of  Bressloff-Cowan's  model}
We  now   present our   conformal modification of  Bressloff and Cowan's spherical  model of a hypercolumn. 
 Let us   first recall the following  two crucial points of   the spherical  model:
\begin{itemize}[leftmargin = 20pt]
\item[(a)] A  hypercolumn $H $ is  identified  with  the  sphere $H_{BC} = S^2$ in such a way that  the north and south pole of $S^2$ are identified with    the  two  pinwheels $\Np$ and $\Sp$ of $H$,  where the parameter $\phi$ takes its extremal values.
\item[(b)] The two parameters $(\theta, \phi)$ for the simple neurons of
 $$\wt H \=  H \setminus \{\Np, \Sp\} = S^2 \setminus \{\text{North pole, South pole}\}$$
  can be  taken as  local coordinates for the RF  $R_H$ of the hypercolumn. In other words,  the projection
$$ \pi: \wt H   \longrightarrow  R_H:= \pi(\wt{H}) \subset  R \ ,   $$
  mapping  each    simple cell  in $\wt H$   to  its  RF  (considered as a point of the  retina $R$) is a local diffeomorphism. 
 This in particular means that {\it the  parameters   $(\theta, \phi)$  cannot be considered as internal parameters, but as  local  coordinates  for  a region of the   retina  $R$} (in fact,  the receptive field of the hypercolumn).
 \end{itemize}
 \par
 At this point, we would like to  remark that  the   spherical  model  of Bressloff and Cowan and   the   symplectic model of Sarti, Citti and Petitot deal with two quite different   aspects  of  the structure of the V1 cortex:  
 \begin{itemize}[leftmargin = 15 pt]
 \item[--] The spherical model is focused on  local aspects (the structure of  hypercolumns)  and in such a model the parameters  $(\theta, \phi)$ are not defined   for   pinwheels;    
 \item[--] The  symplectic model concerns  global aspects of the V1 cortex and   give  mathematical representations just for    the pinwheels. 
 \end{itemize}
We now   propose  two variations   of these  models,  the  {\it conformal  model}  and its associated  {\it reduced model},  which  address    local {\it and}  global features of   the 
 V1 cortex. 
\par\smallskip
Let us begin with the first of the two, the {\it conformal model}. We  assume   that the  firing   of  the   simple   cells  of a hypercolumn  depend not only  on coordinates $(\theta, \phi)$ of a domain  $\cU \subset R$ in the  retina $R$ (as it occurs in Bressloff and Cowan's model),  but  also    on  other  internal coordinates, parameterised by the elements of the similarity group. More precisely,  we  mathematically represent a hypercolumn as   a principal  fiber bundle   $\pi : P \to  R_H \subset R $  over  the RF  $R_H$ of the hypercolumn $H$    with structure group $B_- = \operatorname{Sim}(\bR^2)$ (see \S \ref{sect51} for the definition of $B_- \subset \SL_2(\bC)$).  As we discussed in \S \ref{fund},  this bundle can be locally identified with the Cartan bundle of  the second order conformal  frames
$$ \pi : P =  G  \longrightarrow   S^2 =  G/B_-$$
  of the  conformal group  $G = \SO^o_{1,3} =  \SL_2(\bC)/\bZ_2$. 
\par
\smallskip
 The representation is  based on the following correspondence between points of   the bundle and simple cells of the  hypercolumn $H$: {\it We assume that   each point  $p\in S^2$    corresponds  to  a  family  of simple cells, which is parameterised by the points  of the 
fiber  $G_p = \pi^{-1}(p)$,}  i.e.  by the  four dimensional 
stability subgroup $G_p \subset G$.   In particular, the  north and south poles  of the Bressloff and Cowan's model $H_{BC} =  S^2$ are now corresponding  to the family of  simple cells  parameterised by the elements of the two subgroups
$$   \pi^{-1}(\Sp) =G_{\Sp} = B_-\ ,\qquad \ \pi^{-1}(\Np) = G_{\Np} = B_+\ .$$
In  detail,  we assume the  following  correspondence between the simple cells of  the hypercolumn $H$ and the elements of $G = \SL_2(\bC)$.  Consider  
$$z = x + i y = r e^{i \theta'}\ ,\qquad w = u + i v = \frac{1}{z}$$
as  complex coordinates for   the tangent planes  $T_{\Sp} S^2 \simeq \bC$ and $T_{\Np} S^2 \simeq \bC$,  and use them as   (stereographic) coordinates for  $S^2 \setminus \{\Np\}  = \ster^{-1}_{\Np} (T_{\Sp} S^2)$ and 
  $S^2 \setminus \{\Sp\}  = \ster^{-1}_{\Sp}(T_{\Np} S^2)$,  respectively. \par
  Note that   $r = |z|$ and $\theta' = \arg(z)$ are related with  Bressloff and Cowan's coordinates $(\theta, \phi)$ of  $H_{BC} = S^2$ by             
 \begin{equation} \label{4} r =  2 \tan  (\phi/2), \qquad  \theta'=\theta.  \end{equation}
In particular, the spatial frequency $p$ is related with the  modulus  $r = |z|$ of the stereographic coordinate  by 
\begin{equation}r =  2 \tan  (\phi/2) =  2 \tan  \left( \pi \frac{\log(p/p_-)}{\log(p_+/p_-)}\right)\ .\end{equation}
              Let us now denote by 
$$n_{\Sp}  = (n_{\Sp}^+ ,  n_{\Sp}^-)\qquad \text{and}\qquad n_{\Np} = (n_{\Np}^+, n_{\Np}^-)$$
  the  pairs of (even and odd) neurons, which  work  as the even and odd mother Gabor filters, with receptive profiles (RP) given by the real an imaginary parts of 
 $$\g_{\Sp}^{\bC}(z) := \r_{\Sp} (z) e^{iy} = e^{-1/2 |z|^2 +iy}  \ , \qquad \g_{\Np}^{\bC}(w) := \r_{\Np}(w) e^{iv} = e^{-1/2 |w|^2 +iv}  \ , $$
respectively.  Finally, let us denote by   $S^2_- = \{\phi< 0\}$ and  $S^2_+= \{ \phi > 0\}$  the lower and higher hemispheres of $S^2$, respectively.  \par
\smallskip
{\it  We assume   that each  simple  even (resp.  odd) neuron $n_A = (n_A^+, n_A^-)$, corresponding  to an element}  
$A  =\left( \smallmatrix a & b \\ c & d \endsmallmatrix \right)  \in   \pi^{-1}(S^2_-) $, 
{\it is modelled by   the  real (resp. imaginary) part of the  complex Gabor filter with RP }
\begin{equation}
\begin{split} & \g_A^\bC(z) = A^*(\g_S^\bC(z)) := A^*(\g_S^+(z) + i \g_S^-(z) ) =  |cz+d|^4 e^{ -|z'|^2+ 2i y'}\\
&\hskip 1 cm \text{\it where}\ \  z' = x' + i y' =  A^{-1} z =  \frac{d z - b}{ - c z +a}  \ .
\end{split}
\end{equation}
 In other words, 
  the complex Gabor filter  corresponding to the pair  $n_A = (n_A^+, n_A^-)$ is the filter, which is  obtained  by the action  of   the transformation $A$ on  the complex mother Gabor filter  with  RP $\g_{\Sp}^\bC(z)$.\par
  \smallskip
   {\it Similarly, we assume that  the pair of 
neurons $n_{A'} = (n_{A'}^+, n_{A'}^-)$ associated with the elements  $ A'  = \left( \smallmatrix a' & b' \\ c' & d' \endsmallmatrix \right) \in \pi^{-1}(S^2_+) \subset \SL_2(\bC)$  is associated with the complex Gabor filters with   RP } 
$$\g_{A'}^\bC(w) = A'{}^*(\g_{\Sp'}^\bC(w)) := A'{}^*(\g_{\Sp}^+(w) + i \g_{\Sp}^-(w) ) \ .$$
\par
\smallskip
\noindent{\bf Remark.}  This  model   for hypercolumns offers a physiological realisation of Tits'  presentation of the homogeneous space  $S^2 =  G/ B_\mp$.  We recall that, according to Tits' point of view,   the  points  of a  homogeneous  space $G/K$ are identified with the   subgroups that are conjugated  to a  fixed  subgroup  $K$. In our case   $K= B_\mp \simeq  \Sim(\bR^2)$. 
 Tits'  approach  is  very  important  for the extension of the  theory of  homogeneous  spaces  to the  discrete case. We therefore hope that   it  will have  useful consequences  in neurogeometry.  
 \par
\medskip
\subsection{The reduced model}
At this point it is important to  observe  that, physiologically, the hypercolumn $H$  consists of   a  {\it finite}  number  of  cells. Due to this,  
it is not realistic to parameterise them by  a {\it non-compact}
set, i.e. by the full group $\SL_2(\bC)$.  To tackle this problem, we correct  our original assumptions and assume that $H$ is parameterised by  the elements of a  {\it relatively compact subset $K$} of $G = \SL_2(\bC)$ of the form 
$$K = K^- {\cdot} K^0 {\cdot} K^+\ ,$$
where the sets $ K^\d$, $\d \in \{-, 0, +\}$,  are defined by 
$$K^\d =   \{A \in G^\d\, : \,  \|A - I\|^2 =  \Tr((A-I) \overline{(A-I)}^T) < \r^2\}$$ 
for some fixed  constant $\r > 0$. 
We now need the following theorem. 
\begin{theorem}  Let  $\r>0$ be  a constant such that all elements  $A  =  \left(\smallmatrix 1 & 0\\ c & 1 \endsmallmatrix\right)$ of the compact subset $ K^- \subset G^-$ 
satisfy  $  |c |  < \r$. 
For any given $\ve >0$,  
  the  disc  $\D_\frac{\ve}{1 + \r \ve}\subset \bC \simeq \wh \bC \setminus \{\infty\} $ of radius $ \frac{\ve}{1 + \r \ve} $    and center  $0 (\simeq \Sp)$ is such that   for any  $ A = \left(\smallmatrix 1 & 0\\ c & 1 \endsmallmatrix\right)  \in K^-_\r$ and $z \in \D_{\frac{\ve}{1 + \r \ve}}$
  $$ |Az  - z|    < \ve\ .$$
\end{theorem} 
\begin{proof} It suffices to observe that  if $A \in K^-_\r$ and $z \in \D_\frac{\ve}{1 + \r \ve}$, then 
$$  |Az - z| = \left| \frac{z}{c z + 1} - z\right|= \frac{|c| |z]^2}{|1 - |c||z||} < \ve \left( \frac{\r \ve}{(1 + \r \ve)^2}  \frac{1}{1 - \frac{\r \ve}{1 + \r \ve}} = 
\right) = \ve\ .$$
\ \vskip - 30 pt
\end{proof}
This theorem shows that, according to  our conformal model,  on a sufficiently small neighbourhood  $\mathcal U_{\Sp}$  of the south pinwheel $\Sp \simeq  0$, all 
 elements of the compact set  $K^-$ act  essentially as the identity map.  \par
 \smallskip
 This observation suggest to  consider  
 a {\it reduced model},  that is a model according to which  the neurons corresponding to the set  $\mathcal W = \pi^{-1}(\mathcal U_{\Sp})  \subset  K$,  are parameterised only by  the  four-dimensional domain $K^0 {\cdot} K^+$. 
  In other words, {\it  in our reduced model
   the neurons with RF in   a small neighbourhood $\mathcal U_{\Sp} \subset S^2$
   are  parameterised by the elements  of  the   $2$-dimensional fiber  bundle $\pi_S:  K^0 {\cdot} K^+ \longrightarrow \mathcal U_{\Sp}$. }
   \par
   \smallskip
   Let us now identify the bundle $\pi_{\Sp}:  K^0 {\cdot} K^+ \longrightarrow \mathcal U_{\Sp}$  with a relatively compact portion of 
   the bundle with two dimensional fiber 
  $$\pi_{\Sp} : \bC^*  {\cdot} N^+ \longrightarrow N^+ = \bC^* {\cdot} N^+ /\bC^*  \simeq \bC \simeq S^2 \setminus \{\Np\}\ .$$
  We  recall that 
  \begin{itemize}
  \item  the  group  $B^+ = \bC$ acts  on $ N^+ = \bC^* {\cdot} N^+/\bC^*$  by  parallel translations
 $$z \longmapsto z + b\ ;$$ 
  \item the  group of diagonal matrices $\bC^*$ acts on   $N^+ $ by rotations and homotheties, i.e. transformations  of  the form  $z \to   a z$,  
$a = \lambda e^{i \alpha}$. 
\end{itemize}
  This implies that  $\bC^* {\cdot} N^+$ can be identified with the similarity group $\mathrm{Sim}(\bR^2) $ of $\bR^2 = \bC$ and our 
   bundle $\pi_S: \bC^* {\cdot} N^+ \longrightarrow N^+ \simeq \bR^2$ is  exactly  the SCP bundle   $\pi:   \Sim(\bR^2) \to \bR^2$. \\[10pt]
\par
\smallskip
\subsection{The ``globalised'' reduced model}
 A  reduced model  as  above  (so far  considered   only for  small   neighbourhoods of the pinwheels $\Sp$ and $\Np$) can be given    also for   small neighbourhoods of any other two  antipodal points $\Sp', \Np' \in S^2$.   In fact,  given  an antipodal  pair $\Sp', \Np'$,  in which for instance $\Sp'$ is  in the southern hemisphere,   we may always consider:
  \begin{itemize}[itemsep = 2pt]
  \item[--] a rotation  $A\in \SO_2 \subset \SL_2(\bC)$ such that $\Sp'= A(\Sp)$ and   $\Np' = A(\Np)$
  \item[--] a corresponding new system of spherical coordinates $(\phi', \theta')$ (with new parallels and meridians)  determined  by the  poles  $\Sp', \Np'$. 
  \end{itemize}
 Now,    the     reduced model can be almost verbatim  re-defined for any sufficiently small neighbourhood of $\Sp'$ or $\Np'$, provided that one considers   the new latitude and longitude $(\phi', \theta')$. \par
    \smallskip
 The collection of all    reduced models that can be constructed in this way is parameterised by the set of the antipodal pairs $(\Sp', \Np')$ and it   can be considered as a  {\it globalised  (reduced) model} for the hypercolumn $H$.  According to such globalised model, 
 any  simple neuron corresponding to a  sufficiently small region near a point $\Sp'$ works in term of appropriate new coordinates $(\phi', \theta')$, which  are related  
 with  the original Bressloff and Cowan coordinates  $(\phi, \theta)$  by a  rotation $A \in \SO_2$. 
\par 
                \medskip
  \section{Relations   with  Sarti, Citti and  Petitot's  model and    possible  developments} \label{comparison}
   \subsection{A comparison  between the reduced   model and    Sarti, Citti and  Petitot's  symplectic model} \label{10}
In this section we  compare   Sarti, Citti and Petitot's model of the V1 cortex with our reduced conformal model. 
In fact, both  models  concerns  systems  of simple cells.   More precisely,   Sarti, Citti and Petitot's  model   is about   the simple neurons  of the pinwheels of the  V1 cortex,  while our reduced conformal model is about  the simple neurons  of a (portion of a) hypercolumn.   In both models, the  simple neurons are  identified with  corresponding   Gabor filters,   all of them obtained 
from a  fixed  {\it mother Gabor filter}  using  transformations of a group  $G$.   The group $G$ is the same for both models (it is $G = \Sim(\bR^2)$). However, in our reduced model  such a   group is determined through a reduction of a larger group, namely of  $\SL_2(\bC)$. \par
\smallskip
The construction   of  Gabor filters from a single  mother  filter  leads 
to  a parameterisation  of   neurons  by (a relatively compact subset of) the total space $G$  of the SCP bundle 
\begin{multline*} \pi: G=  \Sim(\bR^2) \longrightarrow \mathcal U \  \text{where}\  \  \mathcal U\ \text{is  a retinal region $R \subset \bR^2$\  } \  \text{(in the symplectic model)} 
 \\ \ \text{or  is an appropriate small domain  of} \\
\text{the  receptive field} \   R_H  
   \ \text{of   a  hypercolumn $H$} \text{ ( in the reduced model)}  \ .
    \end{multline*}
The group $G = \Sim(\bR^2) = \bC^* {\cdot}  \bC$ is $4$-dimensional and, following  Tits' realisation of homogeneous spaces,  each   $z \in \mathcal U$ is  identified with the quotient $G/G_z$  with  $G_z \subset G$ stabiliser of $z$. All stabilisers $G_z \simeq \bC^*$, $z \in \mathcal U$,  are isomorphic  to $\bC^* = \bR^2 \setminus \{0\}$ and, choosing  a fixed stabiliser  $G_{z_o}$, the  bundle can be  trivialised into 
$$ \Sim(\bR^2) \simeq \mathcal U \times G_{z_o} =  \mathcal U \times \bC^* \subset \bC \times \bC^*\ .$$
In  both models, the polar coordinates $( \phi,  \theta)$  for  the stabiliser $G_{z_o} = \bC^* =  \{  \phi e^{i \theta}\}$  admit a double interpretation: 
\begin{itemize}
\item[(a)] They can be considered as  fiber coordinates for the (locally trivialised)  bundle $\pi: \Sim(\bC^2) \simeq \mathcal U \times \bC^* \to \mathcal U$ and, for each point $z\in \mathcal U \subset R$, 
they  parameterise all simple neurons with   RF equal to $z$. According to  this interpretation, they are  {\it internal parameters}. 
\item[(b)]  There is a one-to-one correspondence between 
the elements of   $G_{z_o} = \bC^*$  and  the points of  $\mathcal U \setminus \{z_o\} \subset R$. Indeed,  one may just consider the   natural  
action of $G_{z_o} = \bC^*$ on $  \bC$ and use it to determine a  bijection between   $G_{z_o}$ and the  orbit  $G_{z_o} {\cdot} z' $ of  some  $z'    \neq z_o$. For instance if  we assume  $z_o = 0$ and $z' = 1$ we  have  the  (identity) map 
$$\imath: \bC^* \longrightarrow \bC \setminus \{z_o = 0\} = \bC^* \ ,\qquad \phi e^{i \theta}\overset{\imath} \longmapsto \phi e^{i \theta} {\cdot} 1  
\ .$$  This bijection  allows to consider  the pairs  $(\phi, \theta) $ also as (polar) coordinates on  $\cU \setminus \{z_o\} \subset R$, i.e. as  {\it external parameters}. 
\end{itemize}
 \par \smallskip 
It is however important to remark  that   the symplectic and the reduced model refer to  quite different  physiological objects, namely: 
\begin{itemize}[leftmargin = 10 pt]
\item In  Sarti, Citti and Petitot's model,   the  SCP bundle $\pi:  \Sim(\bR^2)  \to \mathcal U$ represents  the  simple neurons 
of the pinwheels of  V1 cortex. The basis $\mathcal U$ of the bundle  represents the region  of the retina $R$,  given by the RF of all such pinwheels.  
\item In our reduced conformal model,   the SCP bundle  $\pi:  \Sim(\bR^2)   \to \mathcal U$   parameterises the simple neurons (not necessarily belonging to a  pinwheel) corresponding  to a small region of a hypercolumn.  
The base $\mathcal U$ represents  the RF of  the simple neurons, which are  closed either of  the  pinwheels $\Sp$  and
$\Np$ of the hypercolumn.
\end{itemize}
\par
\smallskip
\subsection{Reduced conformal models for the V1 cortex and  a modification of   Sarti, Citti and Petitot's  symplectic  model}
 In this concluding section,  we  would like to indicate how our    reduced model of hypercolumns 
  leads to a natural  extension of    Sarti, Citti and Petitot's symplectic model of the V1 cortex. \par
\smallskip 
Given a  hypercolumn $H$, let us decompose  the corresponding Bressloff and Cowan's sphere $H_{BC} = S^2$ as  the union $S^2 = \Nem\cup \Sem$  of  its     (southern and northern)  hemispheres. 
 According to our reduced model, 
 the systems of simple neurons corresponding to these  two hemispheres
 are  represented by the   bundles $\pi_\pm: \Sim(\bR^2) \simeq \bC^* \times S_\pm  \to S_\pm$.  These two sets of neurons  detect  stimuli of  low spatial frequencies and of  high spatial frequencies, respectively. 
In  Sarti, Citti and Petitot's approach, for each hemisphere we should consider just the  neurons  of  the unique pinwheel ($\Sp$ or $\Np$) of the hemisphere. In this way  the  southern hemisphere $S_-$ (resp. northern hemisphere $S_+$) is represented just  by   the fiber   $\pi^{-1}_-(\Sp) =  \bR^+ {\cdot} \SO_2$ (resp.     $\pi^{-1}(\Np) = \bR^+ {\cdot} \SO_2$).  
\par
\smallskip
This idea,  combined with  the  assumption that  both RF  of  the pinwheels  $\Sp, \Np$ are considered as   points of the retina $R $,  
  leads to a representation of  the  simple neurons of the V1 cortex  by two  (not just one) SCP bundles $\pi: \Sim(\bR^2) \to R$:  one is associated with  the pinwheels detecting  stimuli of low spatial frequencies, 
 the other  with  the pinwheels detecting high frequencies stimuli.   Notice   that the experimental results in  \cite{E-P-G-K-S-K}  support the conjecture of the  existence of  two  independent   systems  for  perceptions of spatial frequencies,  one for the  higher   and another for the lower.   This  is consistent with  the above  model.
\par
\smallskip
Following the same line of arguments,    another  model of the V1 cortex can be proposed. Indeed,  representing each hypercolumn $H$  by the fibers  over its two  pinwheels $\Sp$ and  $\Np$
of   the (non-reduced) conformal bundle $p: \SL_2(\bC) \to S^2$,  i.e. by the fibers 
$$p^{-1}(\Sp) =  N^-{\cdot} \bC^*   = \bC^* \ltimes \bC \ ,\qquad \pi^{-1}(\Np) =   \bC^* {\cdot} N^+ = \bC^* \ltimes \bC \ ,$$
we get  that  the V1 cortex might be represented by two copies of  the  trivial bundle 
$$\widehat  p: \mathbb V = (\bC^* \ltimes \bC) \times R  \to R\ .$$
The main difference between the previous and this  second model is that now, for each point $z$ of the retina, there are   {\it four} (not just {\it two})   internal parameters, say $(\phi, \theta, u, v)$,  two of them corresponding to  the subgroup $\bC^*$ and the other two corresponding to the subgroup $\bC$. The  two parameters $(\phi, \theta)$  for $\bC^*$   can be   identified with 
the internal parameters of the two bundles  $\widehat \pi: \Sim(\bR^2)  \to R$ and are   related with  the spatial frequency and orientation   $(\phi, \theta)$ of the Bressloff  and Cowan spherical model  (see   (b) in \S \ref{10}).  
At this moment, we do not know   possible physiological interpretations for the remaining  two  internal parameters $(u,v)$ for the subgroup $\bC$. Maybe  
they are  related   with   the  derivatives of the orientations and the spatial frequencies of  stimuli.  Another possibility   is   that  these parameters  are  used in   higher level   systems of the  visual system.
\par
\medskip

\section*{Declarations}
D. A.   was partially supported  by  the Grant 
Basis-Foundation ``Leader''   N. 22-7-1-34-1.
 Besides this, no other funds, grants, or
support was received.
\par

\font\smallsmc = cmcsc9
\font\smalltt = cmtt8
\font\smallit = cmti8

\font\smc = cmcsc10 at 12 pt
\font\ssmc = cmcsc10 at 11 pt

\font\sixrm=cmr6
\font\eightrm=cmr8
\renewcommand{\sc}{\sixrm}
\font\cmss=cmss10
\font\cmsss=cmss10 at 7pt
\font\cmssl=cmss10 at 12 pt
%\font\bss=cmssdc10 at 12 pt
\font\bssl=cmssdc10 at 16 pt
\font\cmsslll=cmss10 at 14 pt

\vskip 1.5truecm
\hbox{\parindent=0pt\parskip=0pt
%\hskip  -3truemm
\vbox{\baselineskip 9.5 pt \hsize=3.5truein
\obeylines
{\smallsmc
Dmitri V.  Alekseevsky, 
Higher School of Modern Mathematics -- MIPT,
1 Klimentovskiy per., 
Moscow,
Russia                                                                                     
\&
University of Hradec  Kr\'alov\'e,
Faculty of Science, 
Rokitansk\'eho 62, 
500~03 Hradec Kr\'alov\'e,
Czech Republic
%
%Dmitri Alekseevsky
%
%Institute for Information  Transmission  Problems RAS and
%University of Hradec   Kr\'alov\'e,
%Faculty of Science, Rokitansk\'eho 62, 500~03 Hradec Kr\'alov\'e,  Czech
%Republic
%

}\medskip
{\smallit E-mail}\/: {\smalltt dalekseevsky@iitp.ru}
}
\vbox{\baselineskip 9.5 pt \hsize=3.1truein
\obeylines
{\smallsmc
Andrea Spiro
Scuola di Scienze e Tecnologie
Universit\`a di Camerino
Via Madonna delle Carceri
I-62032 Camerino (Macerata)
Italy
}\medskip
\medskip
{\smallit E-mail}\/: {\smalltt andrea.spiro@unicam.it
}
\ \\
\ \\[5 pt]
}
}

 \end{document}